\documentclass[10pt,a4paper]{article}
\usepackage{amsmath,amsfonts,amsthm,amssymb}
\usepackage[top=3cm,left=3cm,right=2.5cm,bottom=2.5cm]{geometry}

\usepackage{graphicx,color}
\usepackage[usenames,dvipsnames,svgnames,table]{xcolor}
\usepackage[colorlinks = true, linkcolor = blue, urlcolor  = blue, citecolor = violet]{hyperref}
\usepackage{tikz}

\newtheorem{corollary}{Corollary}[section]

\newtheorem{theorem}[corollary]{Theorem}

\theoremstyle{definition}
\newtheorem{remark}[corollary]{Remark}
\newtheorem{definition}[corollary]{Definition}
\newtheorem{example}[corollary]{Example}

\newcommand{\numset}[1]{\mathbb{#1}}
	
	\newcommand{\rr}{\numset{R}}
	
	\newcommand{\zz}{\numset{Z}}
	\newcommand{\nn}{\numset{N}}

	\newcommand{\one}{\mathbf{1}}
	\newcommand{\e}{\mathrm{e}}
		\newcommand{\Exp}[1]{\mathrm{e}^{#1}}

	\makeatletter
	\providecommand*{\diff}%
	{\@ifnextchar^{\DIfF}{\DIfF^{}}}
	\def\DIfF^#1{%
	\mathop{\mathrm{\mathstrut d}}%
	\nolimits^{#1}\gobblespace}
	\def\gobblespace{%
	\futurelet\diffarg\opspace}
	\def\opspace{%
	\let\DiffSpace\!%
	\ifx\diffarg(%
	\let\DiffSpace\relax
	\else
	\ifx\diffarg%
	\let\DiffSpace\relax
	\else
	\ifx\diffarg\{%
	\let\DiffSpace\relax
	\fi\fi\fi\DiffSpace}

	\renewcommand{\d}{\diff}

	\renewcommand{\P}{\mathbb{P}}
	\newcommand{\Q}{\mathbb{Q}}

	\DeclareMathOperator{\supp}{supp}

	\DeclareMathOperator{\tr}{tr}

\newcommand{\cA}{\mathcal{A}}
\newcommand{\cB}{\mathcal{B}}
\newcommand{\cF}{\mathcal{F}}

\newcommand{\cH}{\mathcal{H}}

\newcommand{\cP}{\mathcal{P}}

  \newcommand{\Sc}{s^{\textnormal{cross}}}
  \newcommand{\Sr}{s^{\textnormal{rel}}}

\title{Recurrence times, waiting times \\ and universal entropy production estimators}
\author{Giampaolo Cristadoro\textsuperscript{1}, Mirko Degli Esposti\textsuperscript{2}, Vojkan Jak\v{s}i\'{c}\textsuperscript{3}, Renaud Raqu\'{e}pas\textsuperscript{4}}
\date{}

\begin{document}

\maketitle

\begin{center}
\small
\begin{tabular}{c c c}
   1. Universit\`a degli Studi di Milano-Bicocca & & 2. Universit\`a di Bologna\\
   Dipartimento di Matematica e Applicazioni & & Dipartimento di Fisica e Astronomia ``Augusto Righi'' \\
	 via R.\ Cozzi 55 && via Irnerio 46 \\
   20125 Milano, Italy && 40126 Bologna, Italy \\
	 &&\\
   3. McGill University & & 4. New York University \\
   Department of Mathematics and Statistics & & Courant Institute of Mathematical Sciences \\
   1005--805 rue Sherbrooke~Ouest & & 251 Mercer Street\\
  Montr{\'e}al (Qu{\'e}bec) ~H3A\,0B9, Canada & & New York, NY 10012, United States \\
\end{tabular}
\end{center}

\begin{abstract}
  The universal typical-signal estimators of entropy and cross entropy based on the asymptotics of recurrence and waiting times play an important role in information theory. Building on their construction, we introduce and study universal typical-signal estimators of entropy production in the context of nonequilibrium statistical mechanics of one-sided shifts over finite alphabets.\\

	\noindent\textbf{MSC2020:} \emph{Primary} 82C05, 37B10; \emph{Secondary} 37B20, 37M25, 92D20.
\end{abstract}


\section{Introduction}

The performance studies of the celebrated Lempel--Ziv coding algorithm \cite{LZ77, LZ78} have led to some deep insights into the specific entropy and relative entropy of stationary measures on shift spaces. Notable among those is the characterization of the specific entropy of a stochastic source in terms of the exponential asymptotics of recurrence times of a typical signal, and the related characterization of the specific cross entropy in terms of waiting times~\cite{WZ89,OW93}. Entropic estimators of this type have found diverse practical applications in information theory and related fields; see e.g.~\cite{KASW, GKB08, Ve}.  While the specific entropy and relative entropy are fundamental notions in statistical mechanics, large deviation theory and multifractal formalism~(see e.g.~\cite{OP,vEFS93,CO,Pf, Ge, Cl14}), the aforementioned information-theoretic insights seem to have found only very few theoretical applications in these fields~\cite{ACRV,CR05}.

This note is the first in a series of works dedicated to refinements of the mathematical theory of entropic estimators that originated in information theory, and to their theoretical and practical applications.
The goal of the present note is to illustrate this  research program on one specific problem in statistical mechanics: a typical-signal characterization of entropy production of stationary measures on shift spaces.
Entropy production is a fundamental notion in nonequilibrium statistical mechanics and we will review it briefly in Section­~\ref{ssec:ep}. For the most part, we will focus in this note on the technically simplest case of~$\psi$-mixing systems. In full generality, our results are stated and proved in Section~\ref{sec-general} and  the Appendix, and are further discussed in~\cite{CDEJRa}.
Although these extensions reach further  and are technically more  involved, conceptually, they follow closely the set of ideas introduced  in the $\psi$-mixing case.


The present note is organized as follows. The basic entropic notions are reviewed in Section~\ref{sec-entropies}. Entropy production is
reviewed in Section~\ref{ssec:ep}, where we also state our result under~$\psi$-mixing, Theorem~\ref{basic-thm}.
In Section~\ref{sec-rw-int}, we give a telegraphic overview of entropic estimators based on recurrence and waiting times.
Theorem~\ref{basic-thm} is proven in Section~\ref{sec-proof-basic}.
In Section~\ref{sec-general}, we state and prove several generalizations of Theorem~\ref{basic-thm} which will be further discussed in \cite{CDEJRa}. In Section~\ref{sec-basic-ex}, we describe basic examples to which Theorem~\ref{basic-thm} and its generalizations apply.
Finally, in Section~\ref{sec-remarks}, we briefly discuss some  technical aspects of the proof, and comment on related works that we have learned about in the final stage of completion of this work. In~\cite{CR05}, the same estimator of entropy production was introduced and studied in the context of Gibbs measures for potentials with summable variations. The theoretical results of~\cite{CR05} have been used in~\cite{SGM21,SG21} in numerical computations of entropy production of  DNA sequences.%

\paragraph*{Acknowledgments} This work  was supported by the {\em Agence Nationale de la Recherche} through
the grant NONSTOPS (ANR-17-CE40-0006-01, ANR-17-CE40-0006-02, ANR-17-CE40-0006-03), and was partly developed during VJ’s and MDE's stays at the CY Advanced Studies, whose support is gratefully acknowledged. Another part of this work was done during MDE's stay at McGill University funded by Simons CRM Scholar-in-Residence Program.
Additional funding was provided by the CY Initiative of Excellence (\emph{Investissements d’Avenir}, grant ANR-16-IDEX-0008).
GC acknowledges partial support by the PRIN Grant 2017S35EHN ``Regular and stochastic behaviour in dynamical systems'' of the Italian Ministry of University and Research (MUR),  and by the UMI Group ``DinAmicI''.
VJ acknowledges the support of NSERC. Most of this work was done while RR was a post-doctoral researcher at CY Cergy Paris Universit\'e and supported by the LabEx MME-DII (\emph{Investissements d'Avenir}). Part of this work was also completed during RR's stay at the Centre de recherches math\'ematiques of Universit\'e de Montr\'eal, whose support is gratefully acknowledged.
The authors wish to thank T.\,Benoist and N.\,Cuneo for useful discussions.

\paragraph*{Data availability statement} The datasets analyzed during the current study are publicly available in the Genome Reference Consortium Human Build 38 repository, Patch Release 14~\cite{GRCh38}.

\section{Entropy, relative entropy and cross entropy}
\label{sec-entropies}

Throughout the paper, $\cA$ is a finite set, referred to as an~\emph{alphabet}. Elements of $\cA$ are called~\emph{letters}. Elements of~$\cA^n$ are called~\emph{words} and we use interchangeably the notation $a = (a_1,a_2,\dotsc, a_n)$ and $a = a_1 a_2 \dotsb a_n$.

Let~$\Omega$ be the set~$\cA^{\nn}$ of sequences with values in~$\cA$. Here and in what follows, the set~$\nn$ of natural numbers does not contain~$0$.
We denote a generic element of~$\Omega$ by $x = (x_k)_{k\in\nn}$. We use $x_1^n$ for the $n$-\emph{prefix} of $x$, {i.e.}\ the word $x_1^n := x_1 x_2 \dotsb x_n$. Similarly, $x_{k}^{k+m} := x_k x_{k+1} \dotsb x_{k+m-1}$. To each $a \in \cA^n$ we associate the \emph{basic cylinder}
\[
	[a] := \{ x \in \Omega : x_1^n = a \}.
\]
More generally, to any  $A \subseteq \cA^n$ we associate the subset $[A] = \{x \in \Omega : x_1^n\in A\}$ of~$\Omega$, which we also call a cylinder and denote by~$A$ as well.

We equip $\cA$ with the discrete topology and~$\Omega$  with the corresponding product topology. The set of all Borel probability measures on~$\Omega$ is denoted by~$\cP$ and is equipped with the topology of weak convergence.
The \emph{shift} map is defined by $\sigma : (x_{k})_{k\in\nn} \mapsto (x_{k+1})_{k\in\nn}$, and is  a continuous surjection on~$\Omega$.
The set of shift-invariant elements of $\cP$ is denoted by~$\cP_{\textnormal{inv}}$, and the set of ergodic elements of~$\cP_{\textnormal{inv}}$  by~$\cP_{\textnormal{erg}}$.

The \emph{specific entropy}~$s(\P)$ of~$\P\in \cP_{\textnormal{inv}}$ is defined by the limit
\begin{equation}
	s(\P) := -\lim_{n\to\infty} \frac 1n \sum_{a \in \cA^n} \P_n(a) \log \P_n(a),
\label{lim-en}
\end{equation}
where the logarithm is taken with base~$\Exp{}$ and~$\P_n$ is the $n$-th marginal of~$\P$, i.e.\ the probability measure on~$\cA^n$ defined by
\[
	\P_n(a)
	= \P([a]).
\]
for $a \in \cA^n$. The limit \eqref{lim-en} always exists and lies in~$[0,\log|\cA|]$ by Fekete's lemma for subadditive sequences. The entropy map
$\cP_{\textnormal{inv}}\ni \P \mapsto  s(\P)$ is affine and upper-semicontinuous.

Two other entropic quantities are at the heart of the present article. They both involve two shift invariant measures.
The {\em cross entropy} of $\P \in \cP_{\textnormal{inv}}$ with respect to $\Q \in \cP_{\textnormal{inv}}$ is defined by
\[
	\Sc(\P|\Q) := -\lim_{n\to\infty} \frac 1n \sum_{a \in {\cal A}^n} \P_n(a) \log \Q_n(a)
\]
whenever the limit exists (or the sequence properly diverges as $n\to\infty$), in which case the \emph{relative entropy} of~$\P$ with respect to~$\Q$ is defined as
\[
	\Sr(\P|\Q) := \Sc(\P|\Q) - s(\P).
\]
By an elementary convexity argument, $\Sc(\P|\Q) \geq \Sr(\P|\Q) \geq 0$ whenever well defined. It is known that the cross entropy may fail to exist; see {e.g.}~\cite[\S{A.5.2}]{vEFS93} or Exercise~1.c in~\cite[\S{II.1.e.}]{Shi}.

Relative entropy is also referred to as the Kullback--Leibler divergence and plays a fundamental role in the theory of hypothesis testing. The celebrated Stein lemma gives an operational interpretation of relative entropy in this context;
see \cite[\S{3.4}]{DeZe}, \cite[\S{2.3}]{BJPP18}  for additional information, and~\cite[\S{4.3}]{Jak} for pedagogical introduction and historical perspective to this topic.

\section{Entropy production}
\label{ssec:ep}

Our main interest is in estimating the mean \emph{entropy production} associated to $\P\in \cP_{\textnormal{inv}}$. It is denoted by~$\operatorname{ep}(\P)$ and defined as
\[
	\operatorname{ep}(\P) := \Sr(\P|\widehat{\P})
\]
whenever this relative entropy exists, where $\widehat{\P}\in \cP_{\textnormal{inv}}$ is the \emph{reversal} of the measure~$\P$. This reversal depends on the choice of an involution $\theta : \cA \to \cA$ and is defined by the marginals
\[
	\widehat{\P}_n(a) = \P_n(\widehat{a}),
\]
where the reversal $\widehat{a}$ of the finite word
$a = a_1 a_2 \dotsb a_n \in \cA^n$ is
\[
	\widehat{a} := \theta(a_n) \theta(a_{n-1}) \dotsb \theta(a_1).
\]
The choice of involution is often dictated by the context. A first example of such an involution is of course the identity, denoted~$\operatorname{id}$. If one is interested in the time reversal of a physical systems, some variables (e.g.\ spin) should naturally change sign under time reversal, while some others should not, and the presence of~$\theta$ allows one to take this into account.
Another important example coming from biology is the unique involution $\theta_{\textnormal{Ch}}$ on the alphabet $\{\texttt{C},\texttt{G},\texttt{A},\texttt{T}\}$ such that $\theta_{\textnormal{Ch}}(\texttt{C}) = \texttt{G}$ and $\theta_{\textnormal{Ch}}(\texttt{A}) = \texttt{T}$, which is relevant for the study of Chargaff symmetries in DNA sequences;
see \cite{RKC68,ACDE18} and Remark~\ref{rem-DNA}.
When the choice of~$\theta$ is ambiguous, we use the notation~$\operatorname{ep}(\P;\theta)$ to express the explicit dependence.

Given the general interpretation of relative entropy, $\operatorname{ep}(\P)$ is key to the hypothesis-testing problem for the pair $(\P,\widehat{\P})$, sometimes called ``hypothesis testing of the arrow of time''  when the indices along sequences are interpreted as time variables. In other words, entropy production is a measure of irreversibility of the source giving the outcomes~$x_1, x_2, \dotsc$, $x_n, \dotsc$ For example, for a Markov measure $\P$ coming from a stationary Markov chain $(\pi, P)$ and with $\theta = \operatorname{id}$,
one computes
\[
	\operatorname{ep}(\P) = \sum_{a,b\in\cA} \frac{\pi_a P_{a,b} - \pi_b P_{b,a}}{2} \log \frac{\pi_a P_{a,b}}{\pi_b P_{b,a}},
\]
and notices that~$\operatorname{ep}(\P)$ vanishes if and only if the detailed balance condition $\pi_a P_{a,b} = \pi_b P_{b,a}$ holds.
For further discussions of  entropy production from this general ``hypothesis-testing perspective'', we refer the reader to \cite{JOPS12, BJPP18, CJPS18, CJPS19, BCJP21}.
For  the physics perspective, see  the foundational works \cite{ECM93, GC95a, GC95b, Ma99, LS99} and the reviews \cite{Ru99, JPR11}.

Computing entropy production using the definition of relative entropy requires knowledge of all marginals of~$\P$, information which is often not accessible in practice: think of sequences of measurements coming from a system with some unknown parameters, or of DNA sequences. In this context, motivated by universal lossless data compression algorithms,  it is natural to  look  for a sequence of universal estimators of~$\operatorname{ep}(\P)$ which can be computed as a function of a sample sequence $x$ only. The key to their construction are the recurrence time functions~$R_n$ and~$\widehat R_n$ below.

\begin{definition}\label{rec-tim-def}
	For $n\in\nn$ and $x \in \Omega$,
	\[
	R_n(x) := \inf\left\{k \geq 1 : x_{n+k}^{n+k+n-1} = x_1^n \right\},
	\]
	and
	\[
	  \widehat{R}_n(x) := \inf\left \{ k \geq 1 : x_{n+k}^{n+k+n-1}= \widehat{x_{1}^{n}} \right\}.
	\]
\end{definition}

\begin{remark}
	We have chosen a definition of the recurrence time function~$R_n$ that does not allow for overlaps with the original prefix, as Ornstein and Weiss did in~\cite{OW93}. One could alternatively consider, as does Kontoyiannis in~\cite{Ko98,Ko2}, a definition which does allow for overlaps, i.e.\
	\[
		R'_n(x) := \inf\{k \geq 1 : x_{1+k}^{1+k+n-1} = x_1^n \}.
	\]
	The first choice is more convenient for some of our computations.
	Because there exists $m_n(x) \leq n$ such that ${R}_n(x) = R'_n(\sigma^{m_n(x)-1}(x)) - (n-m_n(x))$, one easily shows that none of the results discussed in this paper  is affected by replacing $R_n$ with~$R'_n$.
\end{remark}

With the convention that $\inf\emptyset = \infty$, both $R_n$ and~$\widehat{R}_n$ take values in
$\nn \cup\{\infty\}$. Taking the logarithm with the convention $\log \infty = \infty$ gives that both $\log R_n$ and~$\log\widehat{R}_n$ take values in $[0,\infty]$. The main idea behind the construction of universal entropy production estimator is to compare how much sooner the $n$-prefix reappears compared to its reversal by looking at the exponential rate of growth of the ratio~$\widehat{R}_n/{R}_n$ as~$n\to\infty$. If the process is reversible (e.g.\ an independent and identically distributed process or a mixing Markov chain satisfying detailed balance), then one expects this ratio to typically not grow exponentially fast with~$n$. On the other hand, if there is a clear direction of time in the underlying process which is sufficiently regular (e.g.\ a mixing Markov chain violating detailed balance), then one expects the reversed prefix to typically appear much later than the original prefix reappears, by a factor which grows exponentially fast with~$n$.

\begin{example}
	{\em Consider~$\cA = \{0,1\}$, the identity involution on~$\cA$, and a sequence
	\[
		x = {0100}{11010}{100110100111010}01001010\dotsc
	\]
	With $n=4$, computing $R_4(x)$ amounts to finding the first reoccurrence of the prefix $x_1^4 = 0100$ which does not overlap; here $R_4(x) = 5$,
	\begin{align*}
		x &= \underline{0100}{\underbrace{1101\underline{0}}_{5}\!\underline{100}110100111010}{010}01010\dotsc
	\end{align*}
	On the other hand, computing $\widehat{R}_4(x)$ amounts to finding the first occurrence of the reversal of that prefix, i.e.~$0010$, which does not overlap with the original prefix; here $\widehat{R}_4(x) = 20$,
	\begin{align*}
		x &= \underline{0100}\!\underbrace{{1101\underline{0}}\underline{100}110100111010}_{20}\hspace{-.6em}\widehat{\phantom{0} 010}01010\dotsc
	\end{align*}
	The reversed prefix takes 4 times as long as the original prefix to appear down the sequence.}
\end{example}

The result at the heart of this note is Theorem~\ref{basic-thm} below, which gives a technically simple and practically important case in which $\tfrac{1}{n}\log {\widehat{R}_n}/{R_n}$ does almost surely converge to~$\operatorname{ep}(\P)$. The hypotheses are formulated in terms of the $\psi$-mixing coefficients recalled below; for generalizations see Section~\ref{sec-general}. We emphasize the {\em universality} aspect of \eqref{ajde}: the sequence of estimators
$\log {\widehat{R}_n}/{R_n}$ is defined deterministically without reference to any random source and the $\P$-almost sure validity of the convergence for a large class of measures~$\P$ is at the essence of this universality. As stated, Theorem~\ref{basic-thm} covers many important examples: mixing Markov and multi-step 
 (a.k.a.~multi-level) Markov measures, mixing hidden Markov models (with finite hidden alphabets), Gibbs measures in the sense of Bowen, and mixing unravelings of quantum instruments; see Section~\ref{sec-basic-ex}. The  generalizations discussed in Section~\ref{sec-general} dispense with mixing altogether.

\begin{theorem}
\label{basic-thm}
	If $\P\in \cP_{\textnormal{inv}}$ is $\psi$-mixing with $\psi^*_\P(0) < \infty$, then
	\begin{equation}\label{ajde}
		\lim_{n\rightarrow \infty} \frac 1n \log\frac{\widehat{R}_n(x)}{R_n(x) } =\operatorname{ep}(\P)
	\end{equation}
	for $\P$-almost all~$x$.
\end{theorem}

\begin{remark} 
\label{pre-remark-ac-1}
	We do not require that $\P_n\ll \widehat \P_n$ for all  $n$. If this relation fails for some $n$, then
	both sides in \eqref{ajde} are equal to $\infty$; see Remark~\ref{remark-ac-1}.
\end{remark}

We recall definition of the \emph{$\psi$-mixing coefficients} of a $\sigma$-invariant measure~$\P$:
\[
	\psi^*_{\P}(\ell) := \sup \left\{\frac{\P([a] \cap \sigma^{-n-\ell}[b]) }{\P([a])\P([b])}: a \in \supp \P_n, n \in \nn, b\in \supp \P_m, m \in\nn \right\}
\]
and
\[
	\psi'_{\P}(\ell) := \inf \left\{\frac{\P([a] \cap \sigma^{-n-\ell}[b]) }{\P([a])\P([b])}: a \in \supp \P_n, n \in \nn, b\in \supp \P_m, m \in\nn \right\}.
\]
These coefficients are respectively nonincreasing and nondecreasing in~$\ell$.
The measure~$\P$ is said to be \emph{$\psi$-mixing} if $\psi^*_{\P}(\ell) \to 1$ and $\psi'_{\P}(\ell) \to 1$ as $\ell \to \infty$.
Note that~$\P$ is $\psi$-mixing if and only if~$\widehat \P$ is $\psi$-mixing. For an excellent review of strong mixing notions, see~\cite{Br05}.%
\footnote{Some remarks are in order to ease comparisons with the setup of Bradley~\cite{Br83,Br05}.
First, note that the coefficients do not change if we replace $[a]$ with~$[A]$, $A \subseteq \cA^n$ and $[b]$ with $[B]$, $B \subseteq \cA^m$:
for example, if $\P([a] \cap \sigma^{-n-\ell}[b]) \leq \psi_{\P}^*(\ell)\P([a])\P([b])$ for all~$a$ and~$b$, then summing over $a \in A$ and~$b\in B$ gives $\P([A] \cap \sigma^{-n-\ell}([B])) \leq \psi_{\P}^*(\ell)\P([A])\P([B])$.
Second, because $m$ is arbitrary, we have a generating semi-algebra at hand and a standard approximation argument shows that we can replace the requirement that $B \subseteq \cA^m$ for some~$m \in \nn$ with the requirement that~$B$ be Borel measurable.
Finally, since we are only interested in $\sigma$-invariant measures on~$\cA^\nn$\,---\,which are naturally in one-to-one correspondence with $\sigma$-invariant measures on~$\cA^\zz$\,---, the definitions then translate to Bradley's definitions on~$\cA^\zz$ exploiting $\sigma$-invariance and yet another approximation procedure by sets now in the semi-algebra built by shifting by~$n$ cylinders naturally associated to sets of the form $A \subseteq \cA^n$ for some~$n$.}

The reader might notice that the proof of Theorem~\ref{basic-thm}\,---\,provided in Section~\ref{sec-proof-basic} and discussed  in Section~\ref{sec-remarks}\,---\,uses only that  $\psi^*_\P(0) < \infty$ and that $\psi'_\P(\ell) > 0$ for some~$\ell \in \nn$. However, the following theorem of Bradley\footnote{The original result, Theorem~1 in~\cite{Br83}, requires the measure~$\P$  to be mixing in the sense of ergodic theory. However, later in the same paper it is  remarked that this extra assumption is superfluous to derive that $\psi'(\ell') > 0$ for some~$\ell'$ implies $\psi'(\ell) \to 1$ as $\ell \to \infty$. As noted by Bradley in his later review~\cite[\S{4.1}]{Br05}, the fact that $\psi'(\ell) \to 0$ in turn implies $\phi$-mixing\,---\,and thus mixing in the sense of ergodic theory\,---\,can be combined with the original result to obtain the variant of the result stated here. According to Bradley in this same review, this version of the result was included in later works at the suggestion of Denker. It is also worth noting that the proof of (any version of) the result relies heavily on the earlier work~\cite{Br80} on $\phi$-mixing.} shows that assuming that~$\P$ is $\psi$-mixing does not represent an additional restriction. In the same article~\cite{Br83}, Bradley provides an example of a $\psi$-mixing measure~$\P$ for which $\psi^*_{\P}(0) = \infty$.
\begin{theorem}[Bradley, 1983]
\label{bradley-thm}
	Let $\P \in \cP_{\textnormal{inv}}$. If there exists $\ell ^*, \ell' \in \nn$ such that $\psi_\P^\ast(\ell^*)<\infty$ and $\psi_\P^\prime(\ell')>0$, then $\P$ is $\psi$-mixing.
\end{theorem}

As discussed in the next section, the $\psi$-mixing condition is commonly used for the waiting-time characterization of entropy and cross entropy.
The requirement that~$\psi^*_\P(0) < \infty$ yields the following upper-decoupling property: there exists a constant~$C$ such that
	\begin{equation}
		\label{basic-dec-12}
 		\P_{n+m} (a b) \leq C \P_n(a) \P_m(b)
	\end{equation}
	for all $n,m \in \nn$, $a \in \cA^n$, and $b \in \cA^m$.
	Since  $\psi^*_\P(0) = \psi^*_{\widehat{\P}}(0)$, the bound~\eqref{basic-dec-12} also holds for~$\widehat \P$.
	The upper-decoupling property of $\widehat \P$ and Fekete's lemma for subadditive sequences give that the limit
	\[\Sc(\P|\widehat \P) := -\lim_{n\to\infty} \frac 1n \sum_{a \in \cA^n} \P_n(a) \log \widehat \P_n(a)\]
	exists; see also Remark~\ref{remark-ac-1}. In particular, the (possibly infinite) entropy production $\operatorname{ep}(\P)$ appearing in Theorem~\ref{basic-thm} is well defined.
	Furthermore, Kingman's subadditive ergodic theorem gives that
	\begin{equation}
	\label{ajde-key}
		\lim_{n\rightarrow \infty}-\frac{1}{n}\log \widehat \P_n(x_1^n)=\Sc(\P|\widehat \P)
	\end{equation}
	for $\P$-almost all~$x$. Such subadditivity arguments, largely absent in the information-theoretic literature, will also play important role in \cite{CDEJRa}.

	The results involving return times are often reformulated in terms of the so-called~``{match lengths}''. The proof of the following corollary follows a standard strategy based on the observation that  ${L}_m(x) \leq n$ if and only if ${R}_n(x) > m$ and that
	$\widehat{L}_m(x) \leq n$ if and only if $\widehat{R}_n(x) > m$.

	\begin{corollary}
	\label{cor:ep-from-matches}
		If $\P\in \cP_{\textnormal{inv}}$ is $\psi$-mixing with $\psi^*_\P(0) < \infty$, then the \emph{match lengths}
		\[
			{L}_m(x) := \sup\{n\in\nn : x_{n+r}^{n+r+n} = {x_1^n} \text{ for some } 1 \leq r \leq m \}
		\]
		and
		\[
			\widehat{L}_m(x) := \sup\{n\in\nn : x_{n+r}^{n+r+n} = \widehat{x_1^n} \text{ for some } 1 \leq r \leq m  \}
		\]
		satisfy
		\begin{equation}
			\lim_{m\to\infty} \left(\frac{\log m}{\widehat{L}_m(x)} - \frac{\log m}{{L}_m(x)}\right) = \operatorname{ep}(\P)
		\end{equation}
		for $\P$-almost all~$x \in\Omega$.
	\end{corollary}

\section{Recurrence and waiting times}
\label{sec-rw-int}

The recurrence time functions  $R_n$ are extensively studied in information theory. We call the \emph{Wyner--Ziv--Ornstein--Weiss theorem} the $\P$-almost sure convergence
\begin{equation}
\label{eq:WZOW}
	\frac{\log R_n}{n} \to h_\P,
\end{equation}
where $h_\P : \Omega \to [0,\infty]$ is the entropy function of the Shannon--McMillan--Breiman theorem. If $\P$ is ergodic, then $h_\P$ is $\P$-almost surely equal to the specific entropy $s(\P)$. The terminology reflects the contributions of Wyner and Ziv~\cite{WZ89} and of Ornstein and Weiss~\cite{OW93}. A particularly elegant proof of~\eqref{eq:WZOW} which significantly influenced our work was given by Kontoyiannis~\cite{Ko98,Ko2}.

Given the basic Wyner--Ziv--Ornstein--Weiss theorem, the proof of
\eqref{ajde} reduces to showing that
\begin{equation}\label{ajde-ver-tr}
\frac{\log \widehat{R}_n}{n} \to \Sc(\P|\widehat \P)
\end{equation}
in the $\P$-almost sure sense.
In turn, the proof of \eqref{ajde-ver-tr} makes use of another
important family of functions  in information theory, the \emph{waiting-time functions}, whose domain consists of pairs~$(x,y)$ of elements of~$\Omega$.
\begin{definition}
	For $n \in \nn$ and $(x, y)\in \Omega \times \Omega$,
	\[
	W_n(x,y) := \inf\{k  \geq 1 : y_{k}^{k+n-1} = x_1^n \}.
	\]
\end{definition}
In other words, $W_n(x,y)$ is the first time the prefix $x_1^n$ of $x$ appears in~$y$. Note that $R_n(x)=W_n(x,\sigma^{n}(x))$.  Just like the recurrence time functions, the waiting-time functions take values in
$\nn \cup\{\infty\}$.

\begin{example}{\em
	Consider~$\cA = \{0,1\}$ and the sequences
	\begin{align*}
		x &= 01001101010011010011101001001010\dotsc, \\
		y &= 11010100010010100110101001001010\dotsc
	\end{align*}
	With $n=4$, computing $W_4(x,y)$ amounts to finding the first occurrence of the prefix $x_1^4 = 0100$ in~$y$; here $W_4(x,y) = 5$ as seen in
	\begin{align*}
		x &= \underline{0100}1101010011010011101001001010\dotsc, \\
		y &= \underbrace{1101\underline{0}}_{5}\!\underline{100}010010100110101001001010\dotsc
	\end{align*}
	}
\end{example}

For our purposes, the relevance of $W_n$ arises through the $(\P\times\Q)$-almost sure convergence
\begin{equation}
	\frac{\log W_n}{n} \to \Sc(\P|\Q),
	\label{basic-cross}
\end{equation}
which is known to hold under  suitable mixing assumptions. In this note, we will make use of the following basic result of Kontoyiannis \cite[\S{4}]{Ko98}; see also \cite[\S{4.2.2}]{Ko2}.
\begin{theorem}[Kontoyiannis, 1998]
\label{thm-mixing-int}
	Suppose that~$\P, \Q \in \cP_{\rm erg}$ with $\P_n \ll \Q_n$ for all~$n \in \nn$. If~$\Q$ is $\psi$-mixing with $\psi_\Q^\ast(0)<\infty$,
	then
	\[
  	\lim_{n\to\infty} \frac{\log W_n(x,y)}{n} = \Sc(\P|\Q)
	\]
	for $(\P\times\Q)$-almost all pairs $(x,y)$.
\end{theorem}
In \cite[\S{4}]{Ko98} this result is given under the  additional assumption that~$\Q$ is  a Markov measure. This assumption, however, is used there only to ensure that the limit
\begin{equation}
\Sc(\P|\Q) = -\lim_{n\to\infty} \frac 1n \sum_{a \in \cA^n} \P_n(a) \log  \Q_n(a)
\label{ajde-cb}
\end{equation}
exists and that
	\begin{equation}
	\label{ajde-key-key}
		\lim_{n\rightarrow \infty}-\frac{1}{n}\log \Q_n(x_1^n)=\Sc(\P|\Q)
	\end{equation}
for $\Q$-almost all~$x$. Hence, it is not needed under the assumption $\psi_\Q^\ast(0)<\infty$, in which case~\eqref{ajde-cb} follows from Fekete's lemma  and~\eqref{ajde-key-key} follows from Kingman's subadditive ergodic theorem.
In a similar spirit, the assumption that $\P_n \ll \Q_n$ for all~$n \in \nn$ can be dropped; see Remark~\ref{remark-ac}. The proof of Theorem~\ref{thm-mixing-int} is discussed in Remark~\ref{remark-wtre}.
The assumptions of Theorem~\ref{thm-mixing-int} can be considerably relaxed; see Section~\ref{sec-general} and~\cite[\S{3}]{CDEJRa}.

\section{Proof of Theorem~\ref{basic-thm}}
\label{sec-proof-basic}

We split the proof into three steps, assuming in accordance with Remark~\ref{pre-remark-ac-1} that $\P_n \ll \widehat \P_n$ for all $n\in \nn$; see Remark~\ref{remark-ac-1}.

\begin{description}
	\item[Step 1: Reduction.]  By the Wyner--Ziv--Ornstein--Weiss theorem, proving Theorem~\ref{basic-thm} reduces to showing that
	\[
	\lim_{n\rightarrow \infty}\frac{1}{n}\log \widehat R_n(x)=\Sc(\P|\widehat \P)
	\]
	for $\P$-almost all $x$.
	{In view of~\eqref{ajde-key}, it suffices to prove that}
	\begin{equation}\label{est-1-ajde}
	\liminf_{n\rightarrow \infty} \left(\frac{1}{n}\log \widehat R_n(x) +\frac{1}{n}\log \widehat \P_n(x_1^n)\right)\geq 0,
	\end{equation}
	\begin{equation}\label{est-2-ajde}
	\limsup_{n\rightarrow \infty} \left(\frac{1}{n}\log \widehat R_n(x) +\frac{1}{n}\log \widehat \P_n(x_1^n)\right)\leq 0,
	\end{equation}
	for $\P$-almost all $x$. For $\epsilon >0$, we set
	\begin{equation}
		B_{n, \epsilon}:=\left\{ x\,:\, \widehat R_n(x)\leq \e^{-\log \widehat \P_n(x_1^n)-n\epsilon}\right\}
		\quad
		\text{ and }
		\quad
		E_{n, \epsilon}:=\left\{ x\,:\, \widehat R_n(x)\geq  \e^{-\log \widehat \P_n(x_1^n)+n\epsilon}\right\}.
	\label{thesame}
	\end{equation}
	By the Borel--Cantelli lemma, \eqref{est-1-ajde} and~\eqref{est-2-ajde} follow if, for every $\epsilon>0$,
	\begin{equation}
	\sum_{n=1}^\infty\P(B_{n, \epsilon})<\infty\qquad\hbox{and}\qquad \sum_{n=1}^\infty \P(E_{n, \epsilon})<\infty.
	\label{ajde-ajde}
	\end{equation}
	The next  two steps are devoted to the proof of \eqref{ajde-ajde}.

	\item[Step 2: The first estimate.]  Let $\epsilon >0$ be arbitrary. Note that
	  \begin{align*}
	    \P( B_{n,\epsilon} ) = \sum_{a \in \supp\P_n} \sum_{j=1}^{\lfloor\Exp{- \log \widehat{\P}_n(a) - n \epsilon}\rfloor}\P\left(\left\{ x :  \widehat {R}_n(x) = j, x_1^n=a \right\}\right),
	  \end{align*}
	where $\supp \P_n$ is the set of all~$a \in \cA^n$ such that $\P_n(a) > 0$.	 Since $\widehat{R}_n(x) = j$ and $x_1^n = a$ imply that $x_1^{2n+j-1}$ is of the form $a\zeta\widehat{a}$ for some $\zeta\in\cA^{j-1}$, we can estimate
	 \begin{equation}
	  \begin{split}
	      \P\left(\left\{ x :  \widehat {R}_n(x) = j, x_1^n=a \right\}\right)
	        &\leq \sum_{\zeta \in \cA^{j-1}} \P_{n+j-1+n} (a \zeta\, \widehat{a}) \\
	        &\leq  C^2\sum_{\zeta \in \cA^{j-1}} \P_{j-1} ( \zeta)\P_n(a)\P_n( \widehat{a})\\
	        &= C^2 \P_n(a)\widehat  \P_{n}(a),
	  \end{split}
	  	        \label{no-le}
	  \end{equation}
	  where we used twice the upper-decoupling property~\eqref{basic-dec-12} and consistency of the marginals. Hence,
	  \begin{align*}
	    \P( B_{n,\epsilon} )
	      &\leq C^2
	      \sum_{a \in \supp\P_n}  \sum_{j=1}^{\lfloor\Exp{- \log \widehat{\P}_n(a) - n \epsilon}\rfloor} \P_n(a) \widehat \P_{n} ({a})  \\
	      &\leq C^2\sum_{a \in \supp\P_n}
	       \Exp{-\log\widehat{\P}_n(a)-n\epsilon} \P_n(a) \widehat \P_n(a)  \\
	      &=C^2\e^{-n\epsilon}.
	  \end{align*}
		The last upper bound on the right-hand side is clearly summable in~$n$, as desired for the first summability condition in~\eqref{ajde-ajde}.

	\item[Step 3: The second estimate.]	Let $\epsilon >0$ be arbitrary. We write
	  \begin{align*}
	    \P(E_{n,\epsilon})
	      &
	      = \sum_{a \in \supp \P_n } \P\left( \left\{ x : \widehat{R}_n(x)  \geq \Exp{-\log\widehat{\P}_n(a) + n\epsilon},\, x_1^{n}=a \right\}\right).
	  \end{align*}
	   Let~$n$ be large enough that $\Exp{\frac 12 n \epsilon}>2$. For  $a\in \supp \P_n$, choose~$m(a)$ so that
	    \begin{equation}
	      \Exp{-\log\widehat{\P}_n(a) + \frac 12 n\epsilon} \leq m(a) - n<\Exp{-\log\widehat{\P}_n(a) +  n\epsilon}.
	  \label{choo-m}
	    \end{equation}
	   Omitting the dependence of $m$ on $a$, we write
	   \begin{equation}
	    \begin{split}
	      \P(E_{n,\epsilon}) &\leq \sum_{a\in \supp \P_n}\sum_{b \in \cA^m} \P\left(\left\{ x   : \widehat {R}_n(x)  > m-n \textnormal{ and } x_{1}^{n+m}=ab \right\}\right) \\
	      &
	      =\sum_{a\in \supp \P_n} \sum_{b \in \cA^m} \P\left(\left\{ x : x_1^{n+m}=ab \textnormal{ and }  b_k^{k+n-1}\not=\widehat a \textnormal{ for } 1\leq k \leq m-n\right\}\right) \\
	      &
	      \leq  C\sum_{a\in \supp \P_n} \sum_{b \in \cA^m} \P_n\times \P_m\left(\left\{ (a, b) : b_k^{k+n-1}\not=\widehat a \textnormal{ for } 1\leq k \leq m-n \right\}\right),
		\end{split}
		\label{mar-late}
		\end{equation}
		
		where the inequality follows from the upper-decoupling property \eqref{basic-dec-12}. Now, using twice that the operation~$\widehat{\,\cdot\,}$ is an involution, and using the definition of~$\widehat{\P}$ in terms of this involution, we write
		\begin{align*}
			&\sum_{a\in \supp \P_n} \sum_{b \in \cA^m} \P_n\times \P_m\left(\left\{ (a, b) : b_k^{k+n-1}\not= \widehat a \textnormal{ for } 1\leq k \leq m-n \right\}\right) \\
				&\qquad\qquad = \sum_{a\in \supp \P_n} \sum_{b \in \cA^m} \P_n\times \P_m\left(\left\{ (a, {b}) : \widehat{b_{k}^{{k}+n-1}} \not = a \textnormal{ for } 1\leq k \leq m-n \right\}\right) \\
				&\qquad\qquad = \sum_{a\in \supp \P_n} \sum_{b \in \cA^m} \P_n\times \P_m\left(\left\{ (a, \widehat{b}) : {b}_{k}^{{k}+n-1}\not = a \textnormal{ for } 1\leq k \leq m-n \right\}\right) \\
			  &\qquad\qquad = \sum_{a\in \supp \P_n} \sum_{b \in \cA^m} \P_n\times \widehat \P_m\left(\left\{ (a, b) : b_{k}^{k+n-1}\not=a\textnormal{ for } 1 \leq k \leq m-n \right\}\right).
		\end{align*}
		The last probability on the right-hand side can be reinterpreted in terms of the waiting times of Section~\ref{sec-rw-int}:
		\begin{align*}
				&\sum_{a\in \supp \P_n} \sum_{b \in \cA^m} \P_n\times \widehat \P_m\left(\left\{ (a, b) : b_{k}^{k+n-1}\not=a \textnormal{ for } 1\leq k \leq m-n  \right\}\right) \\
	      &\qquad\qquad =\sum_{a\in \supp \P_n} \sum_{b \in \cA^m} \P \times \widehat \P\left(\left\{ (x, y) : W_n(x, y)>m-n,\, x_1^n=a, \, y_1^m=b
	      \right\}\right) \\[1mm]
	      & \qquad\qquad \leq \P \times \widehat \P\left( \left\{ (x, y) : W_n(x, y)> \Exp{-\log\widehat{\P}_n(x_1^n) + \frac 12 n\epsilon}\right\}\right).
	 \end{align*}
	 The proof of Theorem~\ref{thm-mixing-int} then gives
	\[
		\sum_{n=1}^\infty  \P \times \widehat \P\left(\left\{ (x, y) : W_n(x, y)> \Exp{-\log\widehat{\P}_n(x_1^n) + \frac 12 n\epsilon} \right\}\right)<\infty,
	\]
	and hence the second summability condition in~\eqref{ajde-ajde}; see Remark~\ref{remark-wt}.
\end{description}

\section{Generalizations}
\label{sec-general}

A careful analysis of the proof of Theorem~\ref{basic-thm} in Section~\ref{sec-proof-basic} and of the accompanying Remarks~\ref{rem-ude} and~\ref{remark-wtre} in Section~\ref{sec-remarks} suggests a clear path to generalizations of Theorem~\ref{basic-thm} beyond the $\psi$-mixing case. The following theorem is our first result in this direction.
\begin{theorem}
\label{thm-abslate-1}
	Let $\P \in \cP_{\textnormal{erg}}$. Suppose that the following hypotheses hold with $o(n)$-sequences $(c_n)_{n\in\nn}$ and $(\tau_n)_{n\in\nn}$ of nonnegative integers:
	\begin{enumerate}
	\item[i.] the upper-decoupling inequalities
	\[
		\P([a] \cap \sigma^{-n-\tau_n}[b]) \leq \Exp{c_n} \P_n(a) \P_m(b)
	\]
	and
	\[
		\widehat{\P}([a] \cap \sigma^{-n-\tau_n}[b]) \leq \Exp{c_n} \widehat{\P}_n(a) \widehat{\P}_m(b)
	\]
 	hold for all~$a \in \cA^n$, $n \in \nn$, $b \in \cA^m$, $m\in\nn$;
	\item[ii.] for all~$\epsilon > 0$, the Kontoyiannis-type estimate
 	\begin{equation*}
		\sum_{n\in\nn}
 			\Exp{c_n} (\P\times\widehat{\P})\left( \left\{ (x,y) :  W_n(x,y)\widehat{\P}_n(x_1^n) > \Exp{n\epsilon}\right\} \right) < \infty
 	\end{equation*}
 	holds.
	\end{enumerate}
	Then,
 \begin{align*}
	 \lim_{n\to\infty} \frac{1}{n} \log \frac{\widehat{R}_n(x)}{R_n(x)}
			= \operatorname{ep}(\P)
 \end{align*}
 for $\P$-almost all~$x$.
\end{theorem}

Furthermore, Hypothesis~i can be relaxed to accommodate an additional error term in the spirit of Bryc and Dembo's work on large deviations~\cite{BD96}, but the validity of~\eqref{ajde-key} must then be postulated separately. Hypothesis~ii can also be relaxed in the same spirit. Pursuing these ideas  leads to the following result.

\begin{theorem}
\label{thm-abslate-2}
	Let $\P \in \cP_{\textnormal{inv}}$. Suppose that, for all~$K \in \nn$ large enough, there exist  $o(n)$-sequences $(c_{n, K})_{n\in\nn}$ and $(\tau_{n, K})_{n\in\nn}$ of nonnegative integers such that the following hypotheses hold:
	\begin{enumerate}
	\item[i'.] the upper-decoupling inequality
 		\[
 			\P([a] \cap \sigma^{-n-\tau_{n, K}}([B])) \leq \Exp{c_{n, K}} \P_n(a) \P_m(B) + \Exp{-Kn}
 		\]
 	holds for all~$a \in \cA$, $n \in \nn$, $B \subseteq \cA^m$, $m\in\nn$;
	\item[i''.] the limit
	\[
		h_{\widehat \P}(x) = \lim_{n\to\infty} \frac{-\log \widehat \P_n(x_1^n)}{n}
	\]
	exists for $\P$-almost all~$x$;
	\item[ii'.] for all~$\epsilon > 0$, the Kontoyiannis-type estimate
 	\begin{equation*}
		\sum_{n\in\nn}
 			\Exp{c_{n, K}} (\P\times\widehat{\P})\left(\left\{ (x,y) : W_n(x,y)\widehat{\P}_n(x_1^n) > \Exp{n\epsilon} \text{ and } \widehat{\P}_n(x_1^n) \geq \Exp{-nK} \right\}\right) < \infty
 	\end{equation*}
 	holds.
	\end{enumerate}
	Then,
	\begin{align*}
		\lim_{n\to\infty} \frac{1}{n} \log \widehat{R}_n(x)
				= h_{\widehat{\P}}(x) 
	\end{align*}
for $\P$-almost all~$x$.
\end{theorem}
\begin{remark} Combined with the Wyner--Ziv--Ornstein--Weiss theorem, this result yields that
\begin{align*}
		\lim_{n\to\infty} \frac{1}{n} \log \frac{\widehat{R}_n(x)}{R_n(x)}
				= h_{\widehat{\P}}(x) - h_\P(x)
	\end{align*}
for $\P$-almost all~$x$ under the same conditions. If in addition $\P\in {\cal P}_{\rm erg}$, then 
 \begin{align*}
	 \lim_{n\to\infty} \frac{1}{n} \log \frac{\widehat{R}_n(x)}{R_n(x)}
			= \operatorname{ep}(\P)
 \end{align*}
 for $\P$-almost all~$x$.
\end{remark}
\begin{remark} If Hypothesis~i' of Theorem \ref{thm-abslate-2} is dropped, the proof still gives that 
\[
\begin{split}
 \liminf_{n\to\infty} \frac{-\log \widehat \P_n(x_1^n)}{n}&\leq 
 \liminf_{n\to\infty} \frac{1}{n} \log \widehat{R}_n(x)\\[1mm]
 &\leq \limsup_{n\to\infty} \frac{1}{n} \log \widehat{R}_n(x)\leq \limsup_{n\to\infty} \frac{-\log \widehat \P_n(x_1^n)}{n}
 \end{split}
\]
 for $\P$-almost all~$x$.
\end{remark} 

The proofs of Theorems  \ref{thm-abslate-1} and \ref{thm-abslate-2} are sketched in the Appendix. The companion paper~\cite{CDEJRa} is devoted to the role of {decoupling inequalities} in establishing Hypothesis~ii or Hypotheses~i'' and~ii'\,---\,both of which are sufficient in order to adapt Kontoyiannis' proof of the convergence of waiting times.  We state here a special case of a result in this direction.

\begin{theorem} 
	Let $\P \in \cP_{\textnormal{inv}}$. Suppose that, for all~$K \in \nn$ large enough, there exist  $o(n)$-sequences $(c_{n, K})_{n\in\nn}$ and $(\tau_{n, K})_{n\in\nn}$ of nonnegative integers such that for each~$a \in \cA^n$ and $B \subseteq \cA^m$ there exists $\ell \leq \tau_{n, K}$ for which
	\[
		\P([a] \cap \sigma^{-n-\ell}([B])) \geq \Exp{-c_{n, K}}\P_n(a)\P_m(B) - \Exp{-Kn}.
	\]
	Then, Hypothesis~ii' of Theorem~\ref{thm-abslate-2} holds for all~$K \in \nn$. If, in addition, the sequences $(c_n)_{n\in\nn}$ and $(\tau_n)_{n\in\nn}$ can be chosen independently of~$K$, then Hypothesis~ii of Theorem~\ref{thm-abslate-1} holds.
\label{thm-abslate-3}
\end{theorem}

Theorem~\ref{thm-abslate-3} is accompanied with suitable generalizations of Kontoyiannis' Theorem~\ref{thm-mixing-int}, and it is in fact these generalizations that are the main results of~\cite{CDEJRa}. They offer a different technical and conceptual perspective on the waiting-time characterization of cross entropies that are rooted in the works~\cite{Ko98, Ko2} and are centered around replacing mixing assumptions with decoupling assumptions. As an illustration, these results allow  to extend the waiting-time characterization of cross entropies to hidden Markov models that are only ergodic, a result which was inaccessible with previous approaches unless $\P=\Q$. 

\section{Examples}
\label{sec-basic-ex}

As stated, Theorem~\ref{basic-thm} covers the following basic examples.

\begin{description}
		\item[Example 1. Markov measures.] Let $\P\in \cP_{\textnormal{inv}}$ be a stationary Markov measure generated by the Markov chain $(\pi, P)$.
		For our purposes, there is no loss of generality in assuming that all entries of the invariant probability vector~$\pi$ are strictly positive. The marginals
		of~$\P$ are given by the formula
		\[
			\P_n(a)=\pi_{a_1}p_{a_1,a_2}\cdots p_{a_{n-1},a_n}.
		\]
		The coefficient~$\psi^*_\P(0)$ can be bounded by the inverse of the smallest entry of~$\pi$.
		The measure~$\P$ is~$\psi$-mixing if and only if the transition matrix~$P$ is irreducible and aperiodic\,---\,or equivalently if, for some $N\in \nn$, all the entries of the matrix~$P^N$ are all strictly positive. By enlarging the alphabet,  multi-step  Markov measures  can be reduced to Markov measures, and so the above applies to them as well.

		\item[Example 2. Hidden Markov measures.] For our purposes, it is convenient to work in the positive-matrix product (PMP) representation of hidden Markov measures. For a discussion of this point of view, see~\cite[\S{2.2}]{BCJP21}. PMP measures are  generated by pairs $(\pi, \{P_a\}_{a\in \cA})$ where~$\pi$ is a $(d\times 1)$ probability vector with strictly positive entries, and $P_a$ is, for each~$a$, a $(d \times d)$ matrix with nonnegative entries such that $P:=\sum_{a\in {\cal A}}P_a$ is a stochastic matrix satisfying $\pi P=\pi$.
		The marginals of the PMP measure~$\P$ generated by $(\pi, \{P_a\})_{a\in \cA}$ are given by the formula
		\[
			\P_n(a)=\pi P_{a_1}\cdots P_{a_n}{\bf 1},
		\]
		where ${\bf 1}$ is the $(1\times d)$ vector with all entries equal to $1$.
		Obviously, $\P\in \cP_{\textnormal{inv}}$.
		The coefficient~$\psi^*_\P(0)$ can be bounded by the square of the inverse of the smallest entry of~$\pi$.
		 Moreover, if~$P$ is irreducible and aperiodic, then~$\P$ is $\psi$-mixing.\footnote{Unlike in the Markov case, this is not a necessary requirement.}

		\item[Example 3. Unravelings of quantum instruments.] Let $\cH$ be a finite-dimensional Hilbert space and~$\cB(\cH)$ the $C^\ast$-algebra of all linear maps $A: \cH \rightarrow \cH$. We denote by $\one$ the identity map in $\cB(\cH)$. A \emph{quantum instrument} on $\cH$ is a pair
		$(\rho, \{\Phi_a\}_{a\in \cA})$, where $\rho\in  \cB(\cH)$ satisfies $\rho >0$ and $\tr(\rho)=1$, and  $\Phi_a: \cB(\cH)\rightarrow \cB(\cH)$, $a\in \cA$, are completely positive linear maps
		such that $\Phi:=\sum_{a\in {\cal A}}\Phi_a$ satisfies $\Phi(\one)=\one$ and $\Phi^\ast(\rho)=\rho$.\footnote{The adjoint $\Phi^\ast$ is defined with respect to the inner product $\langle A, B\rangle=\tr(A^\ast B)$ on $\cB(\cH)$.} The \emph{unraveling} of $(\rho, \{\Phi\}_{a\in \cA})$ is the probability
		measure $\P\in \cP_{\textnormal{inv}}$ defined by the marginals
		\[\P_n(a)=\tr \left(\rho(\Phi_{a_1}\circ \cdots\circ \Phi_{a_n}[\one])\right).
		\]
		For references and detailed discussion of this class of measures from the quantum mechanical perspective, including the study of entropy production, we refer the reader to \cite{BJPP18,BCJP21}.  Although the unravelings can be traced back to the early days of quantum mechanics and have been extensively studied ever since, they have  been rediscovered in~\cite{Ku89} in the unrelated context of fractal analysis; see~\cite{JOP17} for more recent work on this subject and additional references.
		The coefficient~$\psi^*_\P(0)$ is finite by Lemma~3.4 in~\cite{BJPP18}. The map $\Phi$ is called primitive if, for some $N\in \nn$, the  map $\Phi^N$ is positivity improving. If this is the case, one easily shows that  $\psi'_\P(N) > 0$ and the measure $\P$ is~$\psi$-mixing by Bradley's theorem.
		In the context of this class of examples, Theorem~\ref{basic-thm} is an important complement to the works \cite{ BJPP18,BCJP21}.

		\item[Example 4. Gibbs measures.]  The measure $\P \in \cP$ is called a \emph{fully supported Gibbs measure} in the sense of Bowen if there exists a continuous function $F: \Omega \rightarrow \rr$, commonly called a (normalized) \emph{potential}, and a constant $C>0$ such that
		\begin{equation}C^{-1}\e^{-S_n F(x)}\leq \P_n(x_1^n)\leq C \e^{-S_n F(x)}
		\label{gibbs-con}
		\end{equation}
		for all $x\in \Omega$ and all $n\in \nn$, where $S_nF(x) := \sum_{j=0}^{n-1} F(\sigma^j(x))$.
		Gibbs measures play an important role in statistical mechanics and in the theory of dynamical systems. Note that the Gibbs condition~\eqref{gibbs-con} implies that
		\begin{equation}
		\label{basic-dec-1}
		C^{-2} \P_n(x_1^n) \P_m(x_{n+1}^{n+m})\leq  \P_{n+m} (x_1^{n+m}) \leq C^2 \P_n(x_1^n) \P_m(x_{n+1}^{n+m}).
		\end{equation}
		for all $x\in \Omega$ and all $n, m\in\nn$. In particular,
		\[
			0 < \psi'_\P(0) \leq \psi^*_\P(0) < \infty,
		\]
		and Bradley's Theorem~\ref{bradley-thm} yields that every  fully supported Gibbs measure is $\psi$-mixing. This reasoning extends from the case with full support to the case where the support is a topologically mixing subshift of finite type; see~\cite[\S{2}]{Wa05}.
\end{description}

Obviously, any Markov or multi-step Markov measure is a hidden Markov measure. Any  hidden Markov measure is an unraveling of a quantum instruments; see \cite[\S{2.1}]{BCJP21}.  The relation between Gibbs measures and measures described in Examples~2 and~3 is poorly understood.

In the context of Examples~1,~2 and~3, Theorems~\ref{thm-abslate-1} and~\ref{thm-abslate-3} allow to extend the conclusion of Theorem~\ref{basic-thm} to the cases where the stochastic matrix~$P$ and the map~$\Phi$ are only irreducible.\footnote{In the Markov case, irreducibility is equivalent to the ergodicity of $\P$.} Further generalizations involve countably infinite alphabets and the setting of~\cite{BD96}; see~\cite{CDEJRa} for details.

\section{Remarks}
\label{sec-remarks}

\begin{remark}{\bf Absolute continuity in  Theorem~\ref{thm-mixing-int}.}
	Let $\P\in \cP_{\rm erg}$, $\Q \in \cP_{\textnormal{inv}}$. Suppose that there exists $n_0$ and $a \in \cA^{n_0}$ such that  $\Q_{n_0}(a) = 0$ and $\P_{n_0}(a) > 0$.
	Then, for $\P$-almost all~$x$, there exists~$N(x)$ such that $\sigma^{N(x)} (x) \in [a]$ and
	\begin{align*}
		\Q\left(\left\{y : \tfrac 1{N+n_0} \log W_{N+n_0}(x,y) < \infty \right\}\right)
			&\leq \sum_{k=1}^\infty \Q(\sigma^{-k}[x_1^{N+n_0}]) \\
			&\leq \sum_{k=1}^\infty \Q_{n_0}[a] \\
			&= 0
	\end{align*}
	for all~$N \geq N(x)$. Since the above holds for $\P$-almost all~$x$, it follows that
	\begin{align*}
		(\P \times \Q)\left(\left\{ (x,y) : \liminf_{n\to\infty} \frac{\log W_n(x,y)}{n} < \infty \right\}\right)
			&= \int \Q\left(\left\{y : \liminf_{n\to\infty} \frac{\log W_n(x,y)}{n} < \infty \right\}\right) \d\P(x) \\
			&= 0.
	\end{align*}
	On the other hand, if there exists $n_0$ and $a \in \cA^{n_0}$ with $\Q_{n_0}(a) = 0$ and  $\P_{n_0}(a) > 0$, then $\Sc(\P_n|\Q_n) = \infty$ for all~$n\geq n_0$ and $\Sc(\P|\Q) = \infty$ as well.

	In Kontoyiannis' Theorem~\ref{thm-mixing-int}
	it is assumed that $\P_n \ll \Q_n$ for all~$n\in \nn$.  If this assumption fails, then the above discussion gives that
	$\Sc(\P|\Q)=\infty$ and
	 \[
    \lim_{n\to\infty} \frac{\log W_n(x,y)}{n} = \infty
  \]
  for $(\P\times\Q)$-almost all pairs $(x,y)$, and so Theorem~\ref{thm-mixing-int} remains valid.
	\label{remark-ac}
\end{remark}

\begin{remark}{\bf Absolute continuity in  Theorem~\ref{basic-thm}.}
	Related considerations apply to Theorem~\ref{basic-thm}. Suppose that  $\P_n\ll \widehat \P_n$ fails for some $n$. Then,  $\operatorname{ep}(\P)=\Sr(\P|\widehat \P)=+\infty$ and
	\begin{equation}
	\lim_{n\rightarrow \infty}\frac{1}{n}\log \widehat R_n(x)= \infty
	\label{ac-fin}
	\end{equation}
	for $\P$-almost all~$x$.
	Thus in this case Theorem~\ref{basic-thm} holds with both sides in \eqref{ajde} equal to $\infty$.

	To prove that \eqref{ac-fin} holds for $\P$-almost all~$x$, note  first that  \eqref{ajde-key}  holds as a consequence of Kingman's subadditive ergodic theorem which allows for random variables to take values in $[0, \infty]$. Set
		\[ \Omega_0:=\{ x\,:\, \P_n(x_1^n)>0\,\,\hbox{and}\,\, \widehat \P_n(x_1^n)=0\,\, \hbox{for some $n\in {\mathbb N}$} \}.\]
	The estimate \eqref{no-le} gives that, for $\P$-almost all~$x\in \Omega_0$, we have ${\widehat R}_n(x)=\infty$ for all~$n$ large enough.
	The proof of Step~1 gives that
	\[\liminf_{n\rightarrow \infty}\frac{1}{n}\log \widehat R_n(x)\geq \lim_{n\rightarrow \infty}-\frac{1}{n}\log
	\P_n(x_1^n)=\infty\]
	for $\P$-almost all~$x\in \Omega \setminus \Omega_0$. Hence, \eqref{ac-fin} holds for $\P$-almost all~$x$. Note that this argument does not make use of return times.
\label{remark-ac-1}
\end{remark}

\begin{remark}
{\bf Upper decoupling.}
\label{rem-ude}
	The only consequence of the assumption~$\psi^*_\P(0) < \infty$ that was used in the proof  of Theorem~\ref{basic-thm} is the upper-decoupling inequality \eqref{basic-dec-12}. This inequality ensured the existence of the cross entropy $\Sc(\P|\widehat \P)$ and the almost sure convergence expressed by~\eqref{ajde-key}. It was also used in the crucial way in the proof of~\eqref{est-1-ajde} and~\eqref{est-2-ajde} in Steps~2 and~3.
	Thus, it should  not come as a surprise that the key arguments go through with~$\psi^*_\P(0) < \infty$ replaced with $\psi^*_{\P}(\ell^*) < \infty$ for some $\ell^\ast \in \nn$\,---\,which is built in the definition of $\psi$-mixing\,---, provided that one adapts Fekete's lemma and Kingman's theorem to the corresponding upper-decoupling inequality.

	It is this focus on the role of decoupling that allows the generalizations beyond the concepts of mixing discussed in Section~\ref{sec-general}: decoupling properties are postulated in a form that does not involve mixing coefficients. The entire theory is then technically considerably more involved, but, on the positive side, these generalizations allow to reach regimes of applicability of the waiting-time characterization of cross entropy that were previously inaccessible~\cite{CDEJRa}.
\end{remark}

\begin{remark}{\bf Waiting times, cross entropy and mixing.}
\label{remark-wtre}
	The two basic ingredients  of Kontoyiannis' proofs of~\eqref{basic-cross} can be summarized as follows: for all $\epsilon >0$, the sets
	\begin{align*}
		{\cal B}_{n, \epsilon}:=\left\{ (x, y): W_n(x, y)\leq \e^{-\log \Q_n(x_1^n)-n\epsilon}\right\}
		\quad
		\text{and}
		\quad
	 	{\cal E}_{n, \epsilon}:=\left\{ (x,y) : W_n(x,y)\geq  \e^{-\log \Q_n(x_1^n)+n\epsilon}\right\}
	\end{align*}
	satisfy
	\begin{equation}
		\sum_{n=1}^\infty\P\times \Q({\cal B}_{n, \epsilon})<\infty
	\label{ajde-ajde-ko-3}
	\end{equation}
	and
	\begin{equation}
		\sum_{n=1}^\infty \P\times \Q({\cal E}_{n, \epsilon})<\infty,
	\label{ajde-ajde-ko-4}
	\end{equation}
	respectively. The proofs are then  complemented by a separate set of arguments (and assumptions) that ensure the existence of $\Sc(\P|\Q)$ and that
	$\lim_{n\rightarrow \infty}-\tfrac{1}{n}\log \Q_n(x_1^n)=\Sc(\P|\Q)$ for $\P$-almost all $x$; recall the discussion surrounding~\eqref{ajde-key} and \eqref{ajde-cb}--\eqref{ajde-key-key}. It turns out that the proof of
	\eqref{ajde-ajde-ko-3} is very general and involves no mixing assumption.
	The known proofs of~\eqref{ajde-ajde-ko-4} are deeper and critically depend on suitable mixing assumptions. This is reflected in the arguments of Steps~2 and~3. However, the waiting-time results entered only in Step~3 through the validity of \eqref{ajde-ajde-ko-4}\,---\,which in turn is the central ingredient in the proofs of waiting-time characterization of the cross entropy. This observation plays a key role in generalizations of Theorem~\ref{basic-thm} discussed in Section~\ref{sec-general}.

	Under $\psi$-mixing, a proof of~\eqref{ajde-ajde-ko-4} in the style of Kontoyiannis~\cite[\S{2}]{Ko98} goes as follows. Let $\ell \in \nn$ be such that $\psi'_{\Q}(\ell) > 0$.
	For $n \in \nn$ and $a \in \cA^n$ such that $\P_n(a) > 0$, we have
	\begin{align*}
		\P \times \Q(\{(x,y) : W_n(x,y) > t\} \cap [a] \times \Omega)
			&= \P_n(a) \Q\left(\left\{ y : y_k^{k+n-1} \neq a \text{ for all } k \leq t\right\}\right) \\[2mm]
			&\leq \P_n(a) \Q\left(\left\{ y : y_{j(n+\ell) + 1}^{j(n+\ell) + n} \neq a \text{ for all } 0 \leq j < J \right\}\right)
	\end{align*}
	for some natural number~$J$ depending on~$t$ and~$n$ in such a way that
	$ J(n+\ell) \geq t$.
	Let
	\[
		A_j = \left\{y : y_{j'(n+\ell) + 1}^{j'(n+\ell) + n} \neq a \text{ for all } 0 \leq j' < j \right\}
	\]
	and note that, if $\Q(A_j) = 0$ for some $j \leq J$, then the right-hand side of the last estimate vanishes and we need not go further. Assuming that $\Q(A_j) > 0$, we  further estimate
	\begin{align*}
		\P \times \Q(\{(x,y) : W_n(x,y) > t\} \cap [a] \times \Omega)
			&\leq \P_n(a) \Q(A_1) \prod_{j=1}^{J-1} \frac{\Q(A_{j+1})}{\Q(A_j)} \\
			&= \P_n(a) (1 - \Q_n(a)) \prod_{j=1}^{J-1} \left( 1 - \frac{\Q(A_j \cap \sigma^{-j(n+\ell)}[a] )}{\Q(A_j)}\right).
	\end{align*}
	Since
	$A_j \in \mathcal{F}_{jn + (j-1)\ell}$, we may use  $\psi$-mixing coefficients to write
	\begin{align}
	\label{eq:almost-sld}
		\frac{\Q(A_j \cap \sigma^{-j(n+\ell)}[a])}{\Q(A_j)}
		&\geq
		\psi'_{\Q}(\ell) \Q(\sigma^{-j(n+\ell)})
		=
		\psi'_{\Q}(\ell) \Q_n(a).
	\end{align}
	Because we use $\psi'_{\Q}(\ell)$ as a lower bound, we may assume that $\psi'_{\Q}(\ell) < 1$. Now,
	\begin{align*}
		\P \times \Q(\{(x,y) : W_n(x,y) > t\} \cap [a] \times \Omega)
			&\leq \P_n(a) (1-\Q_n(a))(1 - \psi'_{\Q}(\ell)\Q_n(a))^{J-1} \\[1mm]
			&\leq \P_n(a)  (1 - \psi'_\Q(\ell))^{-1}(1 - \psi'_{\Q}(\ell)\Q_n(a))^{J}.
	\end{align*}
	Using the above estimate with $t = \Exp{n\epsilon} \Q_n(a)^{-1}$ and an appropriate~$J$, we find
	\[
		\P \times \Q(\{(x,y) : W_n(x,y)\Q_n(x_1^n) >  \Exp{n\epsilon}\} \cap [a] \times \Omega)
			\leq \P_n(a) (1 - \psi'_\Q(\ell))^{-1} (\psi'_{\Q}(\ell))^{-1}  (n + \ell) \Exp{-n\epsilon}
	\]
	for $n$ large enough.
	We have used the basic inequalities $(1-q)^{1/q} \leq \Exp{-1}$ and $\Exp{-s} \leq s^{-1}$ for $q \in (0,1)$ and $s > 0$ respectively.
	Summing over $a \in \supp \P_n$ yields
	\begin{align*}
		\P\times \Q({\cal E}_{n, \epsilon})
			&\leq  (1 - \psi'_\Q(\ell))^{-1} (\psi'_{\Q}(\ell))^{-1}  (n + \ell) \Exp{-n\epsilon},
	\end{align*}
	and~\eqref{ajde-ajde-ko-4}  follows.

\label{remark-wt}
\end{remark}

\begin{remark}
\label{remark-wcre}
	{\bf Work of Chazottes and Redig.}  When this paper was in the final stage of preparation, we have learned of the work \cite{CR05} where Theorem~\ref{basic-thm} was proven under the assumption that~$\P$ is a Gibbs measure with potential~$F$ of summable variations.
	The proof there follows a completely different strategy and relies on fine upper bounds in the exponential approximation of hitting times for $\psi$-mixing processes obtained previously in~\cite{Ab04, AV09}; see Key-lemmas~1 and~2 in \cite[\S{5}]{CR05}.
	In the same work, these bounds were further used in study of fluctuations (central limit theorem, large deviation principle) of a related class of entropy production estimators.

	The technical and conceptual approach advocated here and in the follow-up work \cite{CDEJRa} is rooted in the program~\cite{JPS}, with the goal of deriving robust theories in terms of underlying assumptions that justify their wide range of practical applications. Our proof of Theorem~\ref{basic-thm} is in a very different spirit from that of Theorem~1(1) of~\cite{CR05}. It emphasizes the decoupling aspect of $\psi$-mixing that is implicit in the proof of Kontoyiannis~\cite[\S{2}]{Ko98},  and which\,---\,once recognized\,---\,allows for far-reaching generalizations. Such generalizations cannot be derived on the basis of the fine estimates established in~\cite{Ab04, AV09}, which are expected to hold only under strong mixing assumptions.  For further discussions of this point of view, we refer the reader to~\cite{CDEJRa} and~\cite{CJPS19}.
\end{remark}

\begin{remark}\label{rem-DNA}
	\textbf{The choice of involution in DNA.} After the publication of~\cite{CR05}, estimators of entropy production in DNA sequences have been numerically computed using recurrence times and longest match lengths~\cite{SGM21,SG21}. To our knowledge, all such computations in the literature correspond to the case where the involution~$\theta$ is equal to the identity on $\{\texttt{C},\texttt{G},\texttt{A},\texttt{T}\}$. With such~$\theta$   the entropy production  measures only the directional irreversibility of the sequence.
	However, in view of conjectures on the so-called ``extended Chargaff symmetry''~\cite{ACDE18}, it is interesting to consider the notion of entropy production that arises from the choice of involution~$\theta_{\text{Ch}}$ already described in Section~\ref{ssec:ep}: $\theta_{\textnormal{Ch}}(\texttt{C}) = \texttt{G}$ and~$\theta_{\textnormal{Ch}}(\texttt{T}) = \texttt{A}$.
	Whereas the original symmetry of~\cite{RKC68} can be cast as the equality $\P_1(a) = \P_1(\theta_{\textnormal{ch}}(a)) = \widehat{\P}_1(a)$ for all $a \in \{\texttt{C},\texttt{G},\texttt{A},\texttt{T}\}$, the extended Chargaff symmetry refers to the stronger conjectured identity $\P = \widehat{\P}$ at the level of the full measures.
	To investigate numerically the validity of this second identity, we estimate~$\operatorname{ep}(\P; \theta_{\text{Ch}})$ using Corollary~\ref{cor:ep-from-matches} in Figure~\ref{fig:ep-DNA}. For the sake of completeness, we compare our results with the analogous estimates for~$\operatorname{ep}(\P; \operatorname{id})$.
\end{remark}


\begin{figure}[h]
	\centering
	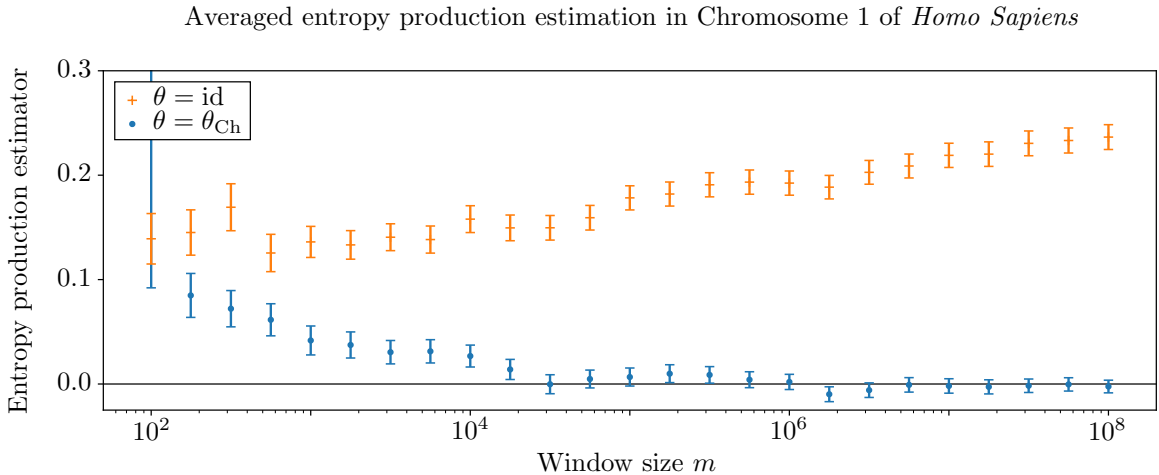
	\vspace{-.5in}
	\caption{In Chromosome~1 of \emph{Homo sapiens}, the estimator of $\operatorname{ep}(\P;\theta)$ of Corollary~\ref{cor:ep-from-matches} can be computed in two important cases: with $\theta = \operatorname{id}$ (orange crosses) and with $\theta = \theta_\textnormal{Ch}$ (blue points).
	Averages over $1\,000$ realizations obtained by choosing uniformly a random initial point in the first half of the sequence are presented together with their estimated standard error of the mean as the window size~$m$ ranges from~$100$ to~$100\,000\,000$. With our choice of the natural logarithm, the $0.1$-unit ticks on the vertical axis correspond to approximately $0.144$ bit per character.
	The sequence used is from the Genome Reference Consortium Human Build 38, Patch Release 14; see~\cite{GRCh38} and~\cite{S+17}.}
\label{fig:ep-DNA}
\end{figure}

\appendix

\section{Appendix}

\begin{proof}[Proof of Theorem \ref{thm-abslate-1}] 
	One follows the steps of the proof of Theorem~\ref{basic-thm} in Section~\ref{sec-proof-basic}, with the following changes. In Step~1, the $\P$-almost sure validity of~\eqref{ajde-key} now relies on Hypothesis~i and adaptations of Fekete's lemma and Kingman's theorem to the corresponding generalized subadditivity condition which are discussed in~\cite{Ra22}. The set $E_{n, \epsilon}$ is defined in the same way, but one takes
	\[B_{n, \epsilon}:=\left\{ x\,:\, \widehat R_{n}(x)\leq \e^{-\log \widehat \P_{n-\tau_n}(x_1^n)-n\epsilon}\right\}.
	\]
	In Step~2, one starts with 
	\begin{align*}
			\P( B_{n,\epsilon} ) = \sum_{a \in \supp\P_n} \sum_{j=1}^{\lfloor\Exp{- \log \widehat{\P}_{n-\tau_n}(a_1^{n-\tau_n}) - n \epsilon}\rfloor}\P\left(\left\{ x :  \widehat {R}_n(x) = j, x_1^n=a \right\}\right),
		\end{align*}
	for $n$ large enough, and estimates
	\begin{equation}
		\begin{split}
			\P(B_{n,\epsilon})
				&\leq  \sum_{a \in \supp\P_n} \sum_{j=1}^{\lfloor\Exp{- \log \widehat{\P}_{n-\tau_n}(a_1^{n-\tau_n}) - n \epsilon}\rfloor}\sum_{\zeta \in \cA^{j-1}} \P_{n+j-1+n} (a \zeta\, \widehat{a}) \\
				&\leq   \sum_{a \in \supp\P_{n-\tau_n}}\sum_{b, b^\prime \in {\cal A}^{\tau_n}} \sum_{j=1}^{\lfloor\Exp{- \log \widehat{\P}_{n-\tau_n}(a) - n \epsilon}\rfloor} \sum_{\zeta \in \cA^{j-1}} \P_{n+j-1+n} (ab \zeta\, b^\prime\widehat{a}) \\[2mm]
				&\leq  \e^{2c_{n-\tau_n}}\sum_{a \in \supp\P_{n-\tau_n}} \sum_{j=1}^{\lfloor\Exp{- \log \widehat{\P}_{n-\tau_n}(a) - n \epsilon}\rfloor}\sum_{\zeta \in \cA^{j-1}} \P_{j-1} ( \zeta)\P_{n-\tau_n}(a)\widehat \P_{n-\tau_n}( {a})\\[2mm]
				&\leq  \e^{2c_{n-\tau_n}}\sum_{a\in\supp\P_{n-\tau_n}}  \Exp{-\log\widehat{\P}_{n-\tau_n}(a)-n\epsilon}\P_{n-\tau_n}(a)  \widehat {\P}_{n-\tau_n}(a)\\[2mm]
				&=\e^{2c_{n-\tau_n}-n\epsilon}.
		\end{split}
	  	\label{no-le-ma}
	  \end{equation}
	  The final estimate and the assumption that $c_n$ and $\tau_n$ are~$o(n)$ give the desired summability condition. 
	  In Step~3, one first chooses, for $n$ large enough, a natural number~$m(a)$ so that 
	  \begin{equation}
	      \Exp{-\log\widehat{\P}_n(a) + \frac 12 n\epsilon} \leq m(a) - n \leq m(a) - n + \tau_n <\Exp{-\log\widehat{\P}_n(a) +  n\epsilon},
	  \label{choo-ma}
		\end{equation}
	    and replaces \eqref{mar-late} with the estimates
	    \[
	    \begin{split}
	      \P(E_{n,\epsilon}) &\leq \sum_{a\in \supp \P_n}\sum_{b^\prime \in {\cal A}^{\tau_n}}\sum_{b \in \cA^{m}} \P \left( \left\{ x   : \widehat {R}_n(x)  > m-n+\tau_n \textnormal{ and } x_{1}^{n+\tau_n+m}=ab^\prime b \right\} \right) \\[2mm]
	      &
	      \leq\sum_{a\in \supp \P_n}\sum_{b^\prime \in {\cal A}^{\tau_n}} \sum_{b \in \cA^{m}} \P\left(\left\{ x : x_1^{n+\tau_n+m}=ab^\prime b \textnormal{ and }  b_k^{k+n-1}\not=\widehat a \textnormal{ for } 1\leq k \leq m-n\right\}\right) \\[2mm]
	      &
	      \leq  \e^{c_n}\sum_{a\in \supp \P_n} \sum_{b \in \cA^m} \P_n\times \P_m\left(\left\{ (a, b) : b_k^{k+n-1}\not=\widehat a \textnormal{ for } 1\leq k \leq m-n \right\}\right).
		\end{split}
		\]
		At this point one proceeds in exactly the same way as in Section \ref{sec-proof-basic} to derive the estimate
		\[
		\P(E_{n, \epsilon})\leq \e^{c_n} \P \times \widehat \P \left( \left\{ (x, y) : W_n(x, y)> \Exp{-\log\widehat{\P}_n(x_1^n) + \frac 12 n\epsilon} \right\} \right),
		\]
	which, combined with Hypothesis~ii, yields the desired summability. 
\end{proof}

\begin{proof}[Proof of Theorem \ref{thm-abslate-2}.] 
	One again follows the same strategy. In Step~1, the $\P$-almost sure validity of~\eqref{ajde-key} now relies on Hypothesis~i''.
	In Step~2, the sets $B_{n, \epsilon}$ are the same as in \eqref{thesame} and one starts with the identity 
	\begin{align*}
		\P( B_{n,\epsilon} ) &= \sum_{a \in \supp\P_n} \sum_{j=1}^{\lfloor \Exp{- \log \widehat{\P}_n(a) - n \epsilon}\rfloor }\P(\{ x :  \widehat{R}_n(x) = j \} \cap [a]),
	\end{align*}
	which gives
	\begin{align*}
		\P( B_{n,\epsilon} ) &\leq \sum_{a \in \supp\P_n} \sum_{j=1}^{\lfloor \Exp{- \log \widehat{\P}_n(a) - n \epsilon}\rfloor }\P([a_1^{n-\tau_{n,K}}] \cap \{ x :  \widehat{R}_n(a\sigma^{n}(x) = j \}) \\
			&= \sum_{a \in \supp\P_n} \P\left(\sigma^{-\tau_{n, K}}[a_1^{n-\tau_{n,K}}] \cap \bigcup_{j=1}^{\lfloor \Exp{- \log \widehat{\P}_n(a) - n \epsilon}
			\rfloor }\sigma^{-\tau_{n, K}}\{ x :  \widehat{R}_n(a\sigma^{n}(x)) = j \}\right)
	\end{align*}
	as soon as $n$ is large enough that $n - \tau_{n,K}  \geq 1$ (this is possible because $\tau_{n,K}$ is~$o(n))$.
	Since each of the $\lfloor \Exp{- \log \widehat{\P}_n(a) - n \epsilon}\rfloor $ sets of the form $\sigma^{-\tau_{n,K}}\{ x :  \widehat{R}_n(a\sigma^{n}(x)) = j \}$ is also necessarily of the form $\sigma^{- n - \tau_{n,K}}(B)$ for some $B \in \cF_{\textnormal{fin}}$ and has probability bounded above by $\widehat{\P}_n(a)$, Hypothesis~i' gives  that 
	\[
	\begin{split}
		\P( B_{n,\epsilon} )
			&\leq \sum_{a \in \supp\P_n} \left( \Exp{c_{n, K}} \P_{n-\tau_{n, K}}(a_1^{n-\tau_{n,K}})  \Exp{- \log \widehat{\P}_n(a) - n \epsilon} \widehat{\P}_n(a) + \Exp{-nK} \right) \\[2mm]
			&\leq \Exp{c_{n, K} + \tau_{n,K} \log|\cA| - n\epsilon}  + \Exp{n\log|\cA|-nK},
	\end{split}
	\]
	where we used $\supp \P_n \subseteq \supp \P_{n-\tau_{n,K}} \times \cA^{\tau_{n,K}}$ to get the second inequality. If  $K > \log|\cA|$, the last estimate gives the desired summability  since both~$c_{n,K}$ and~$\tau_{n,K}$ are~$o(n)$. 
	We now turn to  Step~3. With $E_{n, \epsilon}$ as in \eqref{thesame} and 
	\begin{align*}
			G_{n,K} := \{x : \widehat{\P}([x_1^n] )\geq \Exp{- nK} \},
		\end{align*}
	it suffices to show that  for all~$\epsilon > 0$ and all~$K \in \nn$ large enough, the summability condition
	\begin{equation}
  	\label{eq:ub-summability-cross}
    	\sum_{n = 1}^\infty \P(E_{n,\epsilon} \cap G_{n,K}) < \infty
  	\end{equation}
	holds.\footnote{On the set where   $\tfrac 1n \log \widehat{\P}_n(x_1^n) \to -\infty$ as $n\to\infty$ the inequality $\limsup  \tfrac 1n\log \widehat{R}_n(x) \leq h_{\widehat{\P}}(x)$ is vacuously true (recall Hypothesis~i'').}
	The identity
	\begin{align*}
		\P(E_{n,\epsilon} \cap G_{n,K})
		&
		=  \sum_{a \in \cA^n} \P\left(\left\{ x \in G_{n,K} : \widehat{R}_n(x)  \geq \Exp{-\log\widehat{\P}_n(a) + n\epsilon} \right\} \cap [a]\right)
	\end{align*}
	gives
	%
	\begin{align}
	\label{eq:P-Bn}
		\P(E_{n,\epsilon} \cap G_{n,K})
			\leq \sum_{\substack{a \in \supp \P_n \\ \widehat{\P}_n(a) \geq \Exp{-nK}}}
				\P\left([a] \cap \bigcap_{k=1}^{\lfloor \Exp{-\log\widehat{\P}_n(a) + n\epsilon}  \rfloor -1}\sigma^{-n-k+1} ([\widehat a]^{\mathsf{C}})\right).
	\end{align}
	Setting
	\[
		C_{n, K, \epsilon}(a):=\lfloor \Exp{-\log\widehat{\P}_n(a) + n\epsilon}  \rfloor -\tau_{n,K}-1,
	\]
	and using an obvious inclusion, $\sigma$-invariance, and then  Hypothesis~i', one derives
	\begin{equation}
	\label{eq:cons-UD-ub}
	\begin{split}
		\P\left([a] \cap \bigcap_{k=1}^{\lfloor \Exp{-\log\widehat{\P}_n(a) + n\epsilon}  \rfloor -1}\sigma^{-n-k+1} ([\widehat a]^{\mathsf{C}})\right) 
			&\leq \P\left([a] \cap \bigcap_{k=\tau_{n,K}+1}^{\lfloor \Exp{-\log\widehat{\P}_n(a) + n\epsilon}  \rfloor -1}\sigma^{-n-k+1} ([\widehat a]^{\mathsf{C}})\right) \\
			&\leq \Exp{c_{n,K}} \P([a] )\P\left(\bigcap_{k=1}^{C_{n, K, \epsilon}(a)} \sigma^{-k+1} ([\widehat a]^{\mathsf{C}})\right) + \Exp{-Kn}.
	\end{split}
	\end{equation}
	The identities 
	\begin{align*}
		\P \left(\left\{ y : y_{k}^{k+n-1} \neq \widehat{a}, \, 1\leq k \leq C_{n,K, \epsilon}(a)\right\}\right)
			& =  \P_{C_{n, K, \epsilon}+ n}\left(\left\{ b : b_k^{k+n-1} \neq \widehat{a}, \,1\leq k \leq C_{n,K, \epsilon}(a)\right\} \right)  \\[1mm]
			& =  \P_{C_{n, K, \epsilon} + n}\left(\left\{ \widehat{b} : b_k^{k+n-1} \neq a, \,1\leq k \leq C_{n, K, \epsilon}(a)\right\} \right)  \\[1mm]
			& =  \widehat{\P}_{C_{n, K, \epsilon}+ n}\left(\left\{ b : b_k^{k+n-1} \neq a, \,1\leq k \leq C_{n, k, \epsilon}(a)\right\} \right),
	\end{align*}
	and
	\[
	\begin{split}
		\P([a])\widehat{\P}_{C_{n,K, \epsilon}(a) + n} &\left(\left\{ b : b_k^{k+n-1}\not= a,\,1\leq k \leq C_{n, K, \epsilon}(a)\right\}\right)\\[1mm]
		& = (\P\times\widehat{\P})\left(\left\{(x,y) : W_n(x,y) >  C_{n,K,\epsilon}(x_1^n)\right\} \cap ([a] \times \Omega)\right),
	\end{split}
	\]
	further give
	\begin{equation}
	\begin{split}
		&\P(B_{n,\epsilon} \cap G_{n,K}) \\[2mm]
      		&\qquad \leq \sum_{\substack{a \in \supp \P_n \\ \widehat{\P}_n(a) \geq \Exp{- nK}}} \Exp{c_{n,K}} (\P\times\widehat{\P})\left(\left\{(x,y) : W_n(x,y) > C_{n, K, \epsilon}(x_1^n)\right\} \cap ([a] \times \Omega)\right) + \Exp{-nK} \\[2mm]
	 		&\qquad \leq \Exp{c_{n,K}}(\P\times\widehat{\P})\left(\left\{(x,y) : \widehat{\P}([x_1^n]) \geq \Exp{-nK} \text{ and } W_n(x,y) > C_{n,K, \epsilon}(x_1^n) \right\}\right)+ \Exp{n\log|\cA|-nK}.
	\end{split}
	\label{last-mar}
	\end{equation}
	Note that  for $n$ large enough,
	\begin{equation}
		C_{n, K, \epsilon}(x_1^n) > \Exp{-\log\widehat{\P}([x_1^n]) + \frac 12 n\epsilon}
	\label{late-mar}
	\end{equation}
	for all $x_1^n\in {\cal A}^n$. Taking $K>\log |\cA|$, the estimates~\eqref{last-mar}--\eqref{late-mar} and Hypothesis  ii' give the summability condition~\eqref{eq:ub-summability-cross}.
\end{proof}


\end{document}